\documentclass[smallextended,final,nospthms,natbib]{svjour3} 

\usepackage[english]{babel}
\usepackage{amsmath}
\usepackage{amsthm}
\usepackage{mathpazo}
\usepackage{graphicx}
\usepackage[hyphens]{url}
\usepackage{natbib}
\usepackage[pdf]{pstricks}

\RequirePackage{alltt}
\newenvironment{example}{\begin{alltt}}{\end{alltt}}

\newcommand{\bfx}{\mbox{\boldmath $x$}}
\newcommand{\bfy}{\mbox{\boldmath $y$}}

\newcommand{\bfC}{\mbox{\boldmath $C$}}
\newcommand{\bfR}{\mbox{\boldmath $R$}}
\newcommand{\bfmu}{\mbox{\boldmath $\mu$}}
\newcommand{\bfSi}{\mbox{\boldmath $\Sigma$}}
\newcommand{\bfze}{\mbox{\boldmath $0$}}
\newcommand{\bfxi}{\mbox{\boldmath $\xi$}}
\newcommand{\bfa}{\mbox{\boldmath $a$}}
\newcommand{\bfb}{\mbox{\boldmath $b$}}
\newcommand{\bfp}{\mbox{\boldmath $p$}}
\newcommand{\bfX}{\mbox{\boldmath $X$}}
\newcommand{\bfY}{\mbox{\boldmath $Y$}}
\newcommand{\bft}{\mbox{\boldmath $t$}}

\newcommand{\sbmu}{\mbox{\scriptsize \boldmath $\mu$}}
\newcommand{\sbsi}{\mbox{\scriptsize \boldmath $\Sigma$}}

\def\bfF{\mbox{$\boldsymbol{\mathbf{F}}$}}

\def\bfI{\mbox{$\boldsymbol{\mathbf{I}}$}}

\def\bfU{\mbox{$\boldsymbol{\mathbf{U}}$}}
\def\bfV{\mbox{$\boldsymbol{\mathbf{V}}$}}

\def\bfX{\mbox{$\boldsymbol{\mathbf{X}}$}}

\def\bfSigma{\mbox{$\boldsymbol{\Sigma}$}}
\def\bfOmega{\mbox{$\boldsymbol{\Omega}$}}

\newtheorem{proposition}{Proposition}[section]

\def\ci{\perp\!\!\!\perp} 

\graphicspath{{./fig/}}

\begin{document}

\title{Moments Calculation For the Doubly Truncated Multivariate Normal Density}

\author{Manjunath B G \and Stefan Wilhelm}

\institute{
B. G. Manjunath \at 
CEAUL and DEIO, FCUL, University of Lisbon, Portugal\\
\email{bgmanjunath@gmail.com}
\and
Stefan Wilhelm \at
Department of Finance, University of Basel, Switzerland \\
Tel.: +49-172-3818512\\
\email{Stefan.Wilhelm@stud.unibas.ch}
}


\date{This version: 23.06.2012}

\maketitle

\begin{abstract}
In the present article we derive an explicit expression for the truncated mean and variance for the multivariate normal distribution 
with arbitrary rectangular double truncation. We use the moment generating approach of \cite{Tallis1961} and extend it to general $\bfmu$, $\bfSi$ and 
all combinations of truncation. As part of the solution we also give a formula for the bivariate marginal density of truncated multinormal variates.
We also prove an invariance property of some elements of the inverse covariance after truncation. 
Computer algorithms for computing the truncated mean, variance and the bivariate marginal probabilities for doubly truncated multivariate normal variates have been written in R and are presented along with three examples.
\end{abstract}


\pagebreak

\section{Introduction}

The multivariate normal distribution arises frequently and has a wide range of applications in fields like multivariate regression, Bayesian statistics or the analysis of Brownian motion. One motivation to deal with moments of the truncated multivariate normal distribution comes from the analysis of special financial derivatives (``auto-callables'' or ``Expresszertifikate'') in Germany. These products can expire early depending on some restrictions of the underlying trajectory, if the underlying is above or below certain call levels. In the framework of Brownian motion the finite-dimensional distributions for log returns at any $d$ points in time are multivariate normal. 
When some of the multinormal variates $\bfX=(x_1,\ldots,x_d)' \sim N(\bfmu,\bfSi)$ are subject to inequality constraints (e.g. $a_i \le x_i \le b_i$), this results in truncated multivariate normal distributions.

Several types of truncations and their moment calculation have been described so far, for example the one-sided rectangular 
truncation $\bfx \ge \bfa$ \citep{Tallis1961}, the rather unusual elliptical and radial truncations $\bfa \le \bfx' \bfR \bfx \le \bfb$ \citep{Tallis1963} and the plane truncation $\bfC \bfx \ge \bfp$ \citep{Tallis1965}. Linear constraints like $\bfa \le \bfC \bfx \le \bfb$ can often be reduced to rectangular truncation by transformation of the variables (in case of a full rank matrix $\bfC$ : $\bfa^{*} = \bfC^{-1} \bfa \le \bfx \le \bfC^{-1} \bfb = \bfb^{*}$),
which makes the double rectangular truncation $\bfa \le \bfx \le \bfb$ especially important.

The existing works on moment calculations differ in the number of variables they consider (univariate, bivariate, multivariate) 
and the types of rectangular truncation they allow (single vs. double truncation). 
Single or one-sided truncation can be either from above ($\bfx \le \bfa$) or below ($\bfx \ge \bfa$), but only on one side for all variables,
whereas double truncation $\bfa \le \bfx \le \bfb$ can have both lower and upper truncations points. Other distinguishing features of previous works are further limitations or restrictions they impose on the type of distribution (e.g. zero mean) and the methods they use to derive the results (e.g. direct integration or moment-generating function). Next, we will briefly outline the line of research.

\cite{Rosenbaum1961} gave an explicit formula for the moments of the bivariate case with single truncation from below in both variables by direct integration.
His results for the bivariate normal distribution have been extended by \cite{Shah1964}, \cite{Regier1971} and \cite{Muthen1990} to double truncation.

For the multivariate case, \cite{Tallis1961} derived an explicit expression for the first two moments in case of a singly truncated multivariate normal density with zero mean vector and the correlation matrix $\bfR$ using the moment generating function. \cite{Amemiya1974} and \cite{Lee1979} extended the \cite{Tallis1961} derivation to a general covariance matrix $\bfSi$ and also evaluated the relationship between the first two moments. 
\cite{Gupta1976} and \cite{Lee1983} gave very simple recursive relationships between moments of any order for the doubly truncated case. 
But since except for the mean there are fewer equations than parameters, these recurrent conditions do not uniquely identify moments of order $\ge 2$ and are therefore not sufficient for the computation of the variance and other higher order moments.\\ 
\par
Table \ref{tab:1} summarizes our survey of existing publications dealing with the computation of truncated moments and their limitations.
\begin{table}[tb]
\footnotesize
\caption{Survey of previous works on the moments for the truncated multivariate normal distribution}
\label{tab:1}
\begin{tabular}{p{3cm}lp{1cm}p{5cm}}
\hline\noalign{\smallskip}
 Author & \#Variates & Truncation & Focus \\ 
\noalign{\smallskip}\hline\noalign{\smallskip}
 \cite{Rosenbaum1961} & bivariate    & single & moments for bivariate normal variates with single truncation, $b_1 < y_1 < \infty, b_2 < y_2 < \infty$\\
 \cite{Tallis1961}    & multivariate & single & moments for multivariate normal variates with single truncation from below\\
 \cite{Shah1964}      & bivariate    & double & recurrence relations between moments \\
 \cite{Regier1971}    & bivariate    & double & an explicit formula only for the case of truncation from below at the same point in both variables\\
 \cite{Amemiya1974}   & multivariate & single & relationship between first and second moments \\
 \cite{Gupta1976}     & multivariate & double & recurrence relations between moments \\
 \cite{Lee1979}       & multivariate & single & recurrence relations between moments \\
 \cite{Lee1983}       & multivariate & double & recurrence relations between moments \\
 \cite{Leppard1989}   & multivariate & single & moments for multivariate normal distribution with single truncation\\
 \cite{Muthen1990}    & bivariate    & double & moments for bivariate normal distribution with double truncation, $b_1 < y_1 < a_1, b_2 < y_2 < a_2$ \\
 Manjunath/Wilhelm     & multivariate & double & moments for multivariate normal distribution with double truncation in all variables $\bfa \le \bfx \le \bfb$\\
\noalign{\smallskip}\hline
\end{tabular}
\end{table}
Even though the rectangular truncation $\bfa \le \bfx \le \bfb$ can be found in many situations, no explicit moment formulas for the truncated mean and variance in the general multivariate case of double truncation from below and/or above have been presented so far in the literature and are readily apparent. The contribution of this paper is to derive these formulas for the first two truncated moments and to extend and generalize existing results on moment calculations from especially \cite{Tallis1961,Lee1983,Leppard1989,Muthen1990}.\\
\par
The remainder of this paper is organized as follows. Section 2 presents the moment generating function (m.g.f) for the doubly truncated 
multivariate normal case. In Section 3 we derive the first and second moments by differentiating the m.g.f.
These results are completed in Section 4 by giving a formula for computing the bivariate marginal density. 
In Section 5 we present two numerical examples and compare our results with simulation results. Section 6 links our results to the
theory of graphical models and derives some properties of the inverse covariance matrix.
Finally, Section 7 summarizes our results and gives an outlook for further research. 

\section{Moment Generating Function}

The $d$--dimensional normal density with location parameter vector $\bfmu \in \mathbb{R}^d$
and non-singular covariance matrix $\bfSi$ is given by
\begin{eqnarray} \label{normal}
\varphi_{\sbmu,\sbsi}(\bfx) = \frac{1}{(2\pi)^{d/2} |\bfSi|^{1/2}} \exp\left\{ -\frac{1}{2} \left(\bfx - \bfmu \right)' \bfSi^{-1} \left(\bfx - \bfmu \right) \right\},   &&  \bfx \in \mathbb{R}^d.
\end{eqnarray}
The pertaining distribution function is denoted by $\Phi_{\sbmu,\sbsi} (\bfx)$. Correspondingly, the multivariate truncated normal density, truncated at $\bfa$ and $\bfb$, in $\mathbb{R}^d$, is defined as
\begin{eqnarray} \label{truct. normal}
\varphi{_{\alpha}}_{\sbmu,\sbsi}(\bfx) = \left\{ 
{\begin{array}{*{30}c}
\frac{\varphi_{\sbmu,\sbsi} (\bfx)}{ P\left\{\bfa \leq \bfX \leq \bfb \right\}},  &&  \mbox{  for  } \bfa \leq \bfx \leq \bfb,   \\ 
0, && \mbox{   otherwise. } 
\end{array}} \right.
\end{eqnarray}
Denote $\alpha = P\left\{\bfa \leq \bfX \leq \bfb \right\}$ as the fraction after truncation. 

The moment generating function (m.g.f) of a $d$--dimensional truncated random variable $\bfX$, truncated at $\bfa$ and $\bfb$, in $\mathbb{R}^d$, having the density $f(\bfx)$ is defined as the $d$--fold integral of the form
\begin{eqnarray*}
m(\bft)= E\left(e^{\bft'\bfX}\right) = \int^{\bfb}_{\bfa} e^{\bft'\bfx} f(\bfx) d\bfx.
\end{eqnarray*}

Therefore, the m.g.f for the density in \eqref {truct. normal} is 
\begin{eqnarray} \label{mgf.tn}
m(\bft) = \frac{1}{\alpha(2\pi)^{d/2} |\bfSi|^{1/2}} \int^{\bfb}_{\bfa} \exp\left\{ -\frac{1}{2}   \left[\left(\bfx - \bfmu \right)' \bfSi^{-1} \left(\bfx - \bfmu \right) - 2 \bft'\bfx\right] \right\}  d\bfx.
\end{eqnarray}
In the following, the moments are first derived for the special case $\bfmu = \bfze$. 
Later, the results will be generalized to all $\bfmu$ by applying a location transformation.

Now, consider only the exponent term  in \eqref{mgf.tn} for the case $\bfmu = \bfze$.  Then we have 
\begin{eqnarray*}
-\frac{1}{2}\left[ \bfx ' \bfSi^{-1} \bfx  - 2 \bft'\bfx\right]
\end{eqnarray*}
which can also be written as
\begin{eqnarray*}
\frac{1}{2} \bft'\bfSi \bft - \frac{1}{2} \left[ \left(\bfx  -\bfxi \right)' \bfSi^{-1} \left(\bfx  -\bfxi \right)  \right],
\end{eqnarray*}
where $\bfxi = \bfSi \bft$. \\

Consequently, the m.g.f of the rectangularly doubly truncated multivariate normal is
\begin{eqnarray}
m(\bft) = \frac{e^T}{\alpha(2\pi)^{d/2} |\bfSi|^{1/2}} \int^{\bfb}_{\bfa} \exp\left\{ -\frac{1}{2}   \left[\left(\bfx  -\bfxi \right)' \bfSi^{-1} \left(\bfx - \bfxi \right)  \right] \right\}  d\bfx,
\end{eqnarray}
where $T =  \frac{1}{2} \bft'\bfSi \bft $. \\

The above equation can be further reduced to 
\begin{eqnarray} \label{mgf}
m(\bft) = \frac{e^T}{\alpha(2\pi)^{d/2} |\bfSi|^{1/2}} \int^{\bfb -\bfxi }_{\bfa-\bfxi} \exp\left\{ -\frac{1}{2}    \bfx ' \bfSi^{-1} \bfx  \right\}  d\bfx.
\end{eqnarray}

For notational convenience, we write equation \eqref{mgf} as
\begin{eqnarray} \label{mgfr}
 m(\bft)  = e^T \Phi{_{\alpha}}_{\sbsi}
\end{eqnarray}
where 
\begin{eqnarray*}
 \Phi{_{\alpha}}_{\sbsi} = \frac{1}{\alpha(2\pi)^{d/2} |\bfSi|^{1/2}} \int^{\bfb -\bfxi }_{\bfa-\bfxi} \exp\left\{ -\frac{1}{2}    
\bfx ' \bfSi^{-1} \bfx  \right\}  d\bfx.
\end{eqnarray*}

\section{First And Second Moment Calculation}

In this section we derive the first and second moments of the rectangularly doubly truncated multivariate normal density by differentiating the m.g.f..

Consequently, by taking the partial derivative of (\ref {mgfr}) with respect to $t_i$ we have
\begin{eqnarray} \label {first}
\frac{\partial m(\bft) }{\partial t_i} =  e^T \frac{\partial \Phi{_{\alpha}}_{\sbsi} }{\partial t_i}  +
\Phi{_{\alpha}}_{\sbsi} \frac{\partial e^T }{\partial t_i}.
\end{eqnarray}

In the above equation the only essential terms which will be simplified are 
\begin{eqnarray*}
\frac{\partial e^T }{\partial t_i} =  e^T  \sum^{d}_{k=1} \sigma_{i,k} t_k
\end{eqnarray*}
and
\begin{eqnarray} \label{firstd}
\frac{\partial \Phi{_{\alpha}}_{\sbsi} }{\partial t_i} = \frac{\partial  }{\partial t_i} \int^{b^*_1}_{a^*_1} ... \int^{b^*_d}_{a^*_d} \varphi{_{\alpha}}_{\sbsi}(\bfx) dx_d ... dx_1,
\end{eqnarray}
where $a^*_i = a_i - \sum^{d}_{k=1} \sigma_{i,k} t_k $ and $b^*_i = b_i - \sum^{d}_{k=1} \sigma_{i,k} t_k $.  

Subsequently,  \eqref {firstd} is 
\begin{eqnarray} \label{c1}
\frac{\partial \Phi{_{\alpha}}_{\sbsi} }{\partial t_i} = \sum^{d}_{k=1} \sigma_{i,k} \left(F_k(a^*_k)-F_k(b^*_k)\right), 
\end{eqnarray}
where
\begin{multline}
F_i(x) = \\
         \int^{b^*_1}_{a^*_1}...\int^{b^*_{i-1}}_{a^*_{i-1}}\int^{b^*_{i+1}}_{a^*_{i+1}}...\int^{b^*_{d}}_{a^*_{d}} \varphi{_{\alpha}}_{\sbsi}(x_1,..,x_{i-1},x,x_{i+1},..x_d) dx_d ... dx_{i+1} dx_{i-1} ... dx_1.
\end{multline}         

Note that at $t_k=0$, for all $k=1,2,...,d$, we have $a^*_i = a_i$  and $b^*_i = b_i$. Therefore, $F_i(x)$ will be the $i$--th marginal density. An especially convenient way of computing these one-dimensional marginals is given in \cite{Cartinhour1990}.

From \eqref {first} -- \eqref {c1} for $k=1,2,...,d$ all $t_k = 0$.  Hence, the first moment is
\begin{eqnarray} \label{mean}
E(X_i)= \frac{\partial m(\bft) }{\partial t_i} |_{\bft  =  \bfze} = \sum^{d}_{k=1} \sigma_{i,k} \left(F_k(a_k)-F_k(b_k)\right).
\end{eqnarray}

Now, by taking the partial derivative of  \eqref {first}  with respect to $t_j$, we have
\begin{eqnarray} \label {second}
\frac{\partial^2 m(\bft) }{\partial t_j\partial t_i} =  e^T \frac{\partial^2 \Phi{_{\alpha}}_{\sbsi} }{\partial t_j \partial t_i}  + \frac{\partial \Phi{_{\alpha}}_{\sbsi} }{\partial t_i}  \frac{\partial e^T }{\partial t_j} +   \Phi{_{\alpha}}_{\sbsi} \frac{\partial^2 e^T }{\partial t_j\partial t_i} + \frac{\partial e^T }{\partial t_i} \frac{\partial \Phi{_{\alpha}}_{\sbsi} }{\partial t_j}.
\end{eqnarray}

The essential terms for simplification are
\begin{eqnarray*}
\frac{\partial^2 e^T }{\partial t_j\partial t_i} = \sigma_{i,j} 
\end{eqnarray*}

and clearly, the partial derivative of \eqref{c1} with respect to $t_j$ gives
\begin{eqnarray} \label{c2}
\frac{\partial^2 \Phi{_{\alpha}}_{\sbsi} }{ \partial t_j \partial t_i} = \sum^{d}_{k=1} \left(\sigma_{i,k} \frac{\partial F_k(a^*_k) }{\partial t_j}\right) - \sum^{d}_{k=1} \left(\sigma_{i,k} \frac{\partial F_k(b^*_k) }{\partial t_j}\right).
\end{eqnarray}

In the above equation merely consider the partial derivative of the marginal density $F_k(a^*_k)$ with respect to $t_j$. With further simplification it reduces to
\begin{eqnarray} 
\frac{\partial F_k(a^*_k) }{\partial t_j} &=&  \frac{\partial}{\partial t_j} \int^{b^*_1}_{a^*_1}...\int^{b^*_{k-1}}_{a^*_{k-1}}\int^{b^*_{k+1}}_{a^*_{k+1}}...\int^{b^*_{d}}_{a^*_{d}} \varphi{_{\alpha}}_{\sbsi}(x_1,..,x_{k-1},a^*_k,x_{k+1},..x_d) d\bfx_{-k} \nonumber \\ \label{diff.marg.}
&=&  \frac{ \sigma_{j,k} a^*_k F_k(a^*_k)}{\sigma_{k,k}} \nonumber \\
& &  + \sum_{q \neq k } \left( \sigma_{j,q} - \frac{\sigma_{k,q} \sigma_{j,k}}{\sigma_{k,k}}\right)\left( F_{k,q}(a^*_k,a^*_q) - F_{k,q}(a^*_k,b^*_q) \right), 
\end{eqnarray}
where
\begin{multline} 
\label{sec.marg.}
F_{k,q}(x,y) = \\ \int^{b^*_1}_{a^*_1}...\int^{b^*_{k-1}}_{a^*_{k-1}}\int^{b^*_{k+1}}_{a^*_{k+1}}...\int^{b^*_{q-1}}_{a^*_{q-1}}\int^{b^*_{q+1}}_{a^*_{q+1}}...\int^{b^*_{d}}_{a^*_{d}} \varphi{_{\alpha}}_{\sbsi}(x,y,\bfx_{-k,-q})  d\bfx_{-k, -q}, 
\end{multline}
and the short form $\bfx_{-k}$ denotes the vector $(x_1,..,x_{k-1},x_{k+1},..x_d)'$ in $(d-1)$--dimensions and $\bfx_{-k, -q}$ denotes the $(d-2)$--dimensional vector $(x_1,...,x_{k-1},\\x_{k+1},...,x_{q-1},x_{q+1},...,x_d)'$ for $k \neq q$.
The above equation \eqref {diff.marg.} is deduced from \cite{Lee1979}, pp.~167.  Note that for all $t_k = 0$ the term $F_{k,q}(x , y)$ will be the bivariate marginal density for which we will give a formula in the next section.

Subsequently, $\frac{\partial F_k(b^*_k) }{\partial t_j}$ can be obtained by substituting $a^*_k$ by $b^*_k$. 
From \eqref{second} -- \eqref {sec.marg.} at all $t_k = 0$, $k=1,2,...,d$, the second moment is
\begin{eqnarray} 
\label{sec.moment}
E(X_i X_j) &=& \frac{\partial^2 m(\bft) }{\partial t_j\partial t_i}  |_{\bft = \bfze} \nonumber\\
           &=& \sigma_{i,j} + \sum^{d}_{k=1} \sigma_{i,k}  \frac{ \sigma_{j,k} \left(a_k  F_k(a_k) - b_k F_k(b_k)\right)}{\sigma_{k,k}} \nonumber \\ 
           && + \sum^{d}_{k=1} \sigma_{i,k} \sum_{q \neq k } \left( \sigma_{j,q} - \frac{\sigma_{k,q} \sigma_{j,k}}{\sigma_{k,k}}\right) \left[\left( F_{k,q}(a_k,a_q) - F_{k,q}(a_k,b_q) \right) \right. \nonumber \\
&& \left.- \left(F_{k,q}(b_k,a_q) - F_{k,q}(b_k,b_q)\right)\right].
\end{eqnarray}

Having derived expressions for the first and second moments for double truncation in case of $\bfmu=\bfze$, we will now generalize to all $\bfmu$:\\
if $\bfY \sim N(\bfmu, \bfSi)$ with $\bfa^{*} \le \bfy \le \bfb^{*}$, then $\bfX = \bfY - \bfmu \sim N(\bfze, \bfSi)$ with $\bfa = \bfa^{*} - \bfmu \le \bfx \le \bfb^{*} - \bfmu = \bfb$ and $E(\bfY) = E(\bfX) + \bfmu$ and $Cov(\bfY) = Cov(\bfX)$. Equations \eqref{mean} and \eqref{sec.moment} can then be used to compute $E(\bfX)$ and $Cov(\bfX)$. Hence, for general $\bfmu$, the first moment is
\begin{eqnarray} \label{gmean}
E(Y_i) = \sum^{d}_{k=1} \sigma_{i,k} \left(F_k(a_k)-F_k(b_k)\right) + \mu_i.
\end{eqnarray}
The covariance matrix 
\begin{eqnarray} \label{gvariance}
 Cov(Y_i, Y_j) = Cov(X_i, X_j) = E(X_i X_j) - E(X_i) E(X_j)
\end{eqnarray}
is invariant to the shift in location.

The equations \eqref{gmean} and \eqref{gvariance} in combination with \eqref{mean} and \eqref{sec.moment} 
form our desired result and allow the calculation of the truncated mean and truncated variance for general double truncation.
A formula for the term $F_{k,q}(x_k,x_q)$, the bivariate marginal density, will be given in the next section.\\
We have implemented the moment calculation for mean vector \verb|mean|, covariance matrix \verb|sigma| and truncation vectors \verb|lower| and \verb|upper| as a function 
\begin{verbatim}
 mtmvnorm(mean, sigma, lower, upper)
\end{verbatim}
in the R package \verb|tmvtnorm| \citep{RJournal:Wilhelm+Manjunath:2010,Wilhelm2012}, 
where the code is open source. In Section 5 we show a usage example for this function.

\section{Bivariate Marginal Density Computation}

In order to compute the bivariate marginal density in this section we mainly follow \cite{Tallis1961}, p.~223 and \cite{Leppard1989} who implicitly used the bivariate marginal density as part of the moments calculation for single truncation, evaluated at the integration bounds. 
However, we extend it to the doubly truncated case and state the function for all points within the support region.

Without loss of generality we use a z-transformation for all variates $\bfx = (x_1,\ldots,x_d)'$ as well as 
for all lower and upper truncation points $\bfa = (a_1, \ldots, a_d)'$ and $\bfb = (b_1,\ldots,b_d)'$, resulting in a $N(0,\bfR)$ distribution with correlation matrix $\bfR$ for the standardized untruncated variates. In this section we treat all variables as if they are z-transformed, leaving the notation unchanged.

For computing the bivariate marginal density $F_{q,r}(x_q,x_r)$ with $a_q\le x_q \le b_q, a_r\le x_r \le b_r, q \ne r$, 
we use the fact that for truncated normal densities the conditional densities are truncated normal again. The following relationship holds for $x_s,z_s \in \mathbb{R}^{d-2}$ if we condition on $x_q = c_q$ and $x_r = c_r$ $(s \ne q \ne r)$:
\begin{eqnarray}
  \alpha^{-1} \varphi_d(x_s, x_q = c_q, x_r = c_r; \bfR) & = & \alpha^{-1} \varphi(c_q, c_r; \rho_{qr}) \varphi_{d-2}(z_s; \bfR_{qr}),
\end{eqnarray}
where
\begin{eqnarray}
  z_s & = & (x_s - \beta_{sq.r} c_q - \beta_{sr.q} c_r) / \sqrt{(1-\rho^2_{sq})(1-\rho^2_{sr.q})}
\end{eqnarray}
and $\bfR_{qr}$ is the matrix of second-order partial correlation coefficients for $s \ne q \ne r$. 
$\beta_{sq.r}$ and $\beta_{sr.q}$ are the partial regression coefficients of $x_s$ on $x_q$ and $x_r$ respectively
and $\rho_{sr.q}$ is the partial correlation coefficient between $x_s$ and $x_r$ for fixed $x_q$.\\
Integrating out $(d-2)$ variables $x_s$ leads to $F_{q,r}(x_q,x_r)$ 
as a product of a bivariate normal density $\varphi(x_q,x_r)$ and a $(d-2)$-dimension normal integral $\Phi_{d-2}$:
\begin{eqnarray}
F_{q,r}(x_q=c_q, x_r=c_r)      & = & \int^{b_1}_{a_1}...\int^{b_{q-1}}_{a_{q-1}}\int^{b_{q+1}}_{a_{q+1}}...\int^{b_{r-1}}_{a_{r-1}}     \nonumber \\
 &   & \int^{b_{r+1}}_{a_{r+1}}...\int^{b_{d}}_{a_{d}} \varphi{_{\alpha}}_{R}(x_s, c_q, c_r)  dx_s \nonumber \\
 & = & \alpha^{-1} \varphi(c_q, c_r; \rho_{qr}) \Phi_{d-2}(A^q_{rs}; B^q_{rs}; \bfR_{qr})
\end{eqnarray}
where $A^q_{rs}$ and $B^q_{rs}$ denote the lower and upper integration bounds of $\Phi_{d-2}$ given
$x_q=c_q$ and $x_r=c_r$:
\begin{eqnarray}
   A^q_{rs} & = & (a_s - \beta_{sq.r} c_q - \beta_{sr.q} c_r)/\sqrt{(1-\rho^2_{sq})(1-\rho^2_{sr.q})} \\
   B^q_{rs} & = & (b_s - \beta_{sq.r} c_q - \beta_{sr.q} c_r)/\sqrt{(1-\rho^2_{sq})(1-\rho^2_{sr.q})}.
\end{eqnarray}
The computation of $F_{q,r}(x_q,x_r)$ just needs the evaluation of the normal integral $\Phi_{d-2}$ in ${d-2}$ dimensions, which is readily available in most statistics software packages, for example as the function \verb|pmvnorm()| in the R package \verb|mvtnorm| \citep{Genz2012}.

The bivariate marginal density function 
\begin{example}
dtmvnorm(x, mean, sigma, lower, upper, margin=c(q,r))
\end{example} is also part of the R package \verb|tmvtnorm| \citep{RJournal:Wilhelm+Manjunath:2010,Wilhelm2012}, where readers can find the source code as well as help files and additional examples.

\section{Numerical Examples}

\subsection{Example 1}
We will use the following bivariate example with $\bfmu = (0.5, 0.5)'$ and covariance matrix $\bfSi$
\begin{eqnarray*}
  \bfSi & = & \left(
                  \begin{array}{cc}
                     1   & 1.2 \\
                     1.2 &   2
                  \end{array}
                \right)
\end{eqnarray*}
as well as lower and upper truncation points $\bfa=(-1, -\infty)',\bfb = (0.5, 1)'$, i.e. $x_1$ is doubly, while $x_2$ is singly truncated.
The bivariate marginal density $F_{q,r}(x,y)$ is the density function itself and is shown in figure \ref{fig:bivariate-density}, 
the one-dimensional densities $F_k(x)$ ($k=1,2$) in figure \ref{fig:marginal-densities}.
\begin{figure}
\centering
\includegraphics[width=0.7\textwidth]{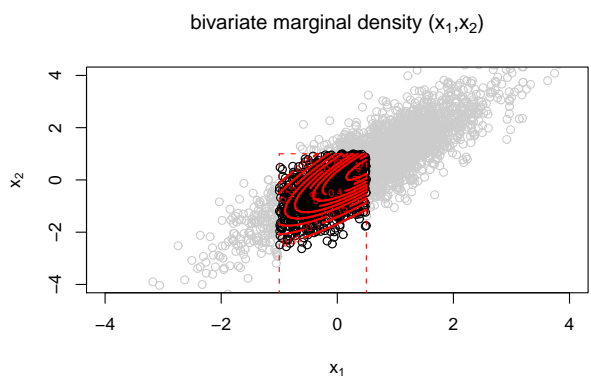}
\caption{Contour plot for the bivariate truncated density function}
\label{fig:bivariate-density}
\end{figure}

\begin{figure*}
\centering
\includegraphics[width=0.7\textwidth]{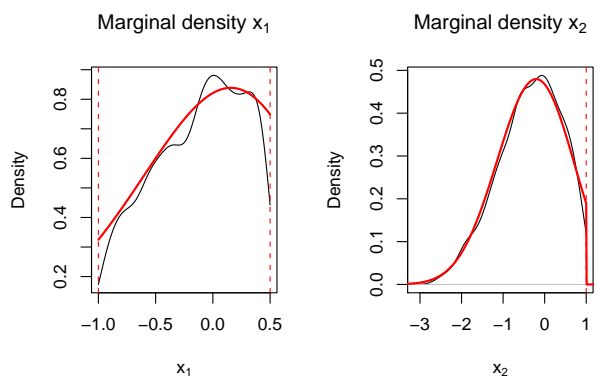}
\caption{Marginal densities $F_k(x)$ ($k=1,2$) for $x_1$ and $x_2$ obtained from Kernel density estimation of random samples and from direct evaluation of $F_k(x)$}
\label{fig:marginal-densities}
\end{figure*}
The moment calculation for our example can be performed in R as
\begin{example}
> library(tmvtnorm)
> mu    <- c(0.5, 0.5)
> sigma <- matrix(c(1, 1.2, 1.2, 2), 2, 2)
> a     <- c(-1, -Inf)
> b     <- c(0.5, 1)
> moments <- mtmvnorm(mean=mu, sigma=sigma, 
>            lower=a, upper=b) 
\end{example}
and results in $\bfmu^{*} = (-0.152, -0.388)'$ and covariance matrix
\begin{eqnarray*}
  \bfSi^{*} & = & \left(
                  \begin{array}{cc}
                     0.163   & 0.161 \\
                     0.161   & 0.606
                  \end{array}
                \right)
\end{eqnarray*}
The trace plots in figures \ref{fig:MC-Mean} and \ref{fig:MC-Cov} show the evolution of a Monte Carlo estimate for the elements of the mean vector and the covariance matrix respectively for growing sample sizes. Furthermore, the 95\% confidence interval obtained from Monte Carlo using the full sample of 10000 items is shown. All confidence intervals contain the true theoretical value, but Monte Carlo estimates still show substantial variation even with a sample size of 10000. Simulation from a truncated multivariate distribution
and calculating the sample mean or the sample covariance respectively also leads to consistent estimates of $\bfmu^{*}$ and $\bfSi^{*}$. Since the rate of convergence of the MC estimator is $O(\sqrt{n})$, one has to ensure sufficient Monte Carlo iterations in order to have a good approximation or to choose variance reduction techniques.
\begin{figure*}
\centering
\includegraphics[width=0.7\textwidth]{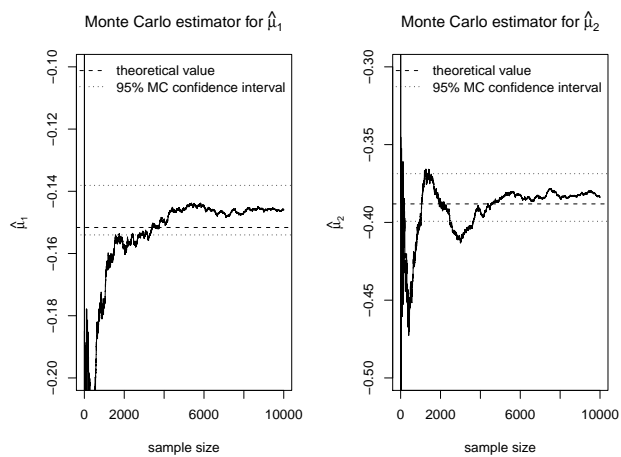}
\caption{Trace plots of the Monte Carlo estimator for $\bfmu^{*}$}
\label{fig:MC-Mean}
\end{figure*}

\begin{figure*}
\centering
\includegraphics[width=0.7\textwidth]{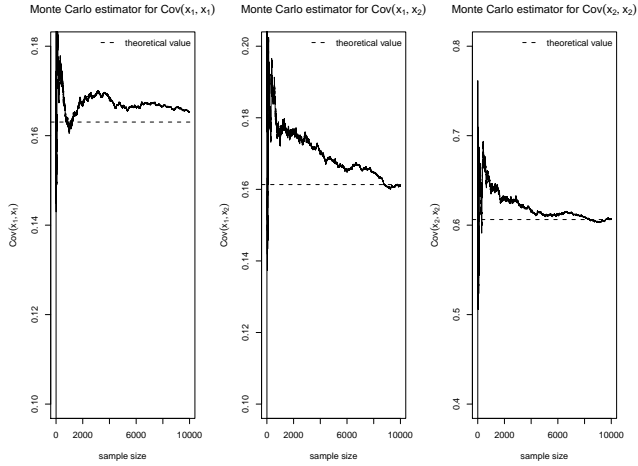}
\caption{Trace plots of the Monte Carlo estimator for the 3 elements of $\bfSi^{*}$ ($\sigma_{11}^{*}$, $\sigma_{12}^{*} = \sigma_{21}^{*}$ and $\sigma_{22}^{*}$)}
\label{fig:MC-Cov}
\end{figure*}

\subsection{Example 2}
Let $\bfmu = (0, 0, 0)'$,the covariance matrix
\begin{eqnarray*}
  \bfSi & = & \left(
                  \begin{array}{rrr}
                     1.1 &  1.2 &    0 \\
                     1.2 &    2 & -0.8 \\
                       0 & -0.8 &    3
                  \end{array}
                \right)
\end{eqnarray*}
and the lower and upper truncation points $\bfa=(-1, -\infty, -\infty)'$ and $\bfb = (0.5, \infty, \infty)'$, then the only truncated variable is $x_1$, 
which is furthermore uncorrelated with $x_3$.\\
Our formula results in $\bfmu^{*} = c(-0.210, -0.229, 0)'$ and
\begin{eqnarray*}
  \bfSi^{*} & = & \left(
                  \begin{array}{rrr}
                     0.174 & 0.190 &  0.0 \\
                     0.190 & 0.898 & -0.8 \\
                         0 & -0.8  &  3.0 
                  \end{array}
                \right)
\end{eqnarray*}
For this special case of only $k < d$ truncated variables $(x_1, \ldots, x_k)$, the remaining $d-k$ variables $(x_{k+1}, \ldots, x_d)$ can be regressed on the truncated variables, and a simple formula for the mean and covariance matrix can be given (see \cite{Johnson1971}, p.\,70).\\ 
Let the covariance matrix $\bfSi$ of $(x_1,\ldots, x_d)$ be partitioned as
\begin{eqnarray}
   \bfSi & = & 
     \left(
     \begin{array}{cc}
	     \bfV_{11}  & \bfV_{12} \\
       \bfV_{21}  & \bfV_{22}
    \end{array}
   \right)
\end{eqnarray}
where $\bfV_{11}$ denotes the $k \times k$ covariance matrix of $(x_1, \ldots, x_k)$. 
The mean vector\footnote{
The formula for the truncated mean given in \cite{Johnson1971}, p.\,70 is only valid for a zero-mean vector or after demeaning all variables appropriately. For non-zero means $\bfmu=(\bfmu_1, \bfmu_2)'$ it will be $(\boldsymbol{\xi'_1}, \bfmu_2 + (\boldsymbol{\xi'_1} - \bfmu_1) \bfV^{-1}_{11} \bfV_{12})$.} and the covariance matrix $\bfSigma^*$ of all $d$ variables can be computed as 
\begin{eqnarray}
  (\boldsymbol{\xi'_1}, \boldsymbol{\xi'_1} \bfV^{-1}_{11} \bfV_{12})
  \label{eqn:1}
\end{eqnarray}
and 
\begin{eqnarray}
   \bfSigma^* & = & 
   \left(
     \begin{array}{rr}
	     \bfU_{11}                            & \bfU_{11} \bfV^{-1}_{11} \bfV_{12} \\
       \bfV_{21} \bfV^{-1}_{11} \bfU_{11}   & \quad \bfV_{22} - \bfV_{21} (\bfV^{-1}_{11} - \bfV^{-1}_{11} \bfU_{11} \bfV^{-1}_{11}) \bfV_{12}
    \end{array}
   \right)
   \label{eqn:2}
\end{eqnarray}
where $\boldsymbol{\xi'_1}$ and $\bfU_{11}$ are the mean and covariance of the $(x_1, \ldots, x_k)$ after truncation.\\
The mean and standard deviation for the univariate truncated normal $x_1$ are
\begin{eqnarray*}
  \xi_1 = \mu_1^{*}    & = & \sigma_{11} \frac{\varphi_{\mu_1,\sigma_{11}}(a_1) - \varphi_{\mu_1,\sigma_{11}}(b_1)}{ \Phi_{\mu_1,\sigma_{11}}(b_1) - \Phi_{\mu_1,\sigma_{11}}(a_1)} \\
  \sigma_{11}^{*}      & = & \sigma_{11} + \sigma_{11} \frac{a_1 \varphi_{\mu_1,\sigma_{11}}(a_1) - b_1 \varphi_{\mu_1,\sigma_{11}}(b_1)}{\Phi_{\mu_1,\sigma_{11}}(b_1) - \Phi_{\mu_1,\sigma_{11}}(a_1)}
  \label{eqn:3}
\end{eqnarray*}
Letting $\bfU_{11} = \sigma_{11}^{*}$ and inserting $\xi_1$ and $\bfU_{11}$ into equations \eqref{eqn:1} and \eqref{eqn:2}, one can verify our formula and the results for $\bfmu^{*}$ and $\bfSi^{*}$. 
However, the crux in using the Johnson/Kotz formula 
is the need to first compute the moments of the truncated variables $(x_1, \ldots, x_k)$ for $k \ge 2$. But this has been exactly the subject of our paper.

\section{Moment Calculation and Conditional Independence}

In this section we establish a link between our moment calculation and the theory of graphical models (\cite{Whittaker1990}, \cite{Edwards1995} and \cite{Lauritzen1996}).
We present some properties of the inverse covariance matrix and show how the dependence structure of variables is affected after selection.\\
\par  
Graphical modelling uses graphical representations of variables as nodes in a graph and dependencies among them as edges. 
A key concept in graphical modelling is the conditional independence property. 
Two variables $X$ and $Y$ are conditional independent given a variable or a set of variables $Z$ (notation $X \ci Y | Z$),
when $X$ and $Y$ are independent after partialling out the effect of $Z$. 
For conditional independent $X$ and $Y$ the edge between them in the graph is omitted and 
the joint density factorizes as $f(x, y | z) = f(x | z) f(y | z)$.\\
\par
Conditional independence is equivalent to having zero elements $\bfOmega_{xy}$ 
in the inverse covariance matrix $\bfOmega = \bfSigma^{-1}$ as well as having
a zero partial covariance/correlation between $X$ and $Y$ given the remaining variables:
\[
 X \ci Y | \text{Rest}  \iff \bfOmega_{xy} = 0 \iff \rho_{xy.Rest} = 0
\]
Both marginal independence and conditional independence between variables 
simplify the computations of the truncated covariance in equation \eqref{sec.moment}. 
In the presence of conditional independence of $i$ and $j$ given $q$, the terms $\sigma_{ij} - \sigma_{iq} \sigma^{-1}_{qq} \sigma_{qj} = 0$ vanish as they reflect the partial covariance of $i$ and $j$ given $q$.
\par
As has been shown by \cite{Marchetti2008}, the conditional independence property is preserved after selection, i.e. the inverse covariance matrices 
$\bfOmega$ and $\bfOmega^*$ before and after truncation share the same zero-elements.\\
We prove that many elements of the precision matrix are invariant to truncation.
For the case of $k < d$ truncated variables, we define the set of truncated variables with $T=\{x_1, \ldots, x_k\}$,
and the remaining $d-k$ variables as $S=\{x_{k+1}, \ldots, x_d\}$. 
We can show that the off-diagonal elements $\bfOmega_{i,j}$ are invariant after truncation for $i \in T \cup S$ and $j \in S$: 
\begin{proposition}
The off-diagonal elements $\bfOmega_{i,j}$ and the diagonal elements $\bfOmega_{j,j}$ are invariant after truncation for $i \in T \cup S$ and $j \in S$.
\end{proposition}

\begin{proof}
The proof is a direct application of the Johnson/Kotz formula in equation \eqref{eqn:2} in the previous section. As a result of the formula for partitioned inverse matrices (\cite{Greene2003}, section A.5.3), the corresponding inverse covariance matrix $\bfOmega$
of the partitioned covariance matrix $\bfSigma$ is
\begin{eqnarray}
\bfOmega & = & 
\left(
\begin{array}{rr}
\bfV_{11}^{-1} ( \bfI + \bfV_{12} \bfF_2 \bfV_{21} \bfV_{11}^{-1} ) & \quad -\bfV_{11}^{-1} \bfV_{12} \bfF_2  \\
                          -\bfF_2 \bfV_{21} \bfV_{11}^{-1} &                      \bfF_2                               
\end{array}
\right)
\end{eqnarray}
with $\bfF_2 = (\bfV_{22} - \bfV_{21} \bfV_{11}^{-1} \bfV_{12})^{-1}$.\\
\par
Inverting the truncated covariance matrix $\bfSigma^*$ in equation \eqref{eqn:2} using the formula for the partitioned inverse leads to 
the truncated precision matrix
\begin{eqnarray}
\bfOmega^* = 
   \left(
     \begin{array}{rr}
	    \bfU_{11}^{-1} + \bfV_{11}^{-1} \bfV_{12} \bfF_2 \bfV_{21} \bfV_{11}^{-1}  & \quad -\bfV_{11}^{-1} \bfV_{12} \bfF_2 \\
      -\bfF_2 \bfV_{21} \bfV_{11}^{-1}                           & \quad \bfF_2
    \end{array}
   \right)
   \label{eqn:truncated-precision}
\end{eqnarray}
where the $\bfOmega^*_{12}$ and $\bfOmega^*_{21}$ elements are the same as $\bfOmega_{12}$ and $\bfOmega_{21}$ respectively.
The same is true for the elements in $\bfOmega^*_{22}$, especially the diagonal elements in $\bfOmega^*_{22}$.
\end{proof}
Here, we prove this invariance property only for a subset of truncated variables. 
Based on our experiments we conjecture that the same is true also for the case of full truncation (i.e. all off-diagonal elements in $\bfOmega^*_{11}$), but we do not give a rigorous proof here and leave it to future research.

\subsection{Example 3}
We illustrate the invariance of the elements of the inverse covariance matrix with the famous
mathematics marks example used in \cite{Whittaker1990} and \cite{Edwards1995}, p.\,49.
The independence graph of the five variables $(W, V, X,$ \allowbreak $Y, Z)$ in this example 
takes the form of a butterfly.\\

\begin{center}
\begin{pspicture}(6,2)
\psset{dotsize=6pt,fillstyle=none}

\psline{o-o}(0,0)(0,2)
\psline{o-o}(0,2)(3,1)
\psline{o-o}(0,0)(3,1)

\psline{o-o}(6,0)(6,2)
\psline{o-o}(6,0)(3,1)
\psline{o-o}(6,2)(3,1)

\rput(0,-0.3){mechanics (V)}
\rput(0,2.3){vectors (W)}
\rput(3,0.55){algebra (X)}
\rput(6,-0.3){analysis (Y)}
\rput(6,2.3){statistics (Z)}

\end{pspicture}
\end{center}

Here, we have the conditional independencies $(W, V) \ci (Y, Z) | X$. 
A corresponding precision matrix might look like (sample data; zero-elements marked as "."):
\begin{eqnarray}
\bfOmega & = &
\left(
\begin{array}{rrrrrr}
	    1 &  0.2 &  0.3 &   . & .   \\
    0.2 &    1 & -0.1 &   . & .   \\
    0.3 & -0.1 &    1 & 0.4 & 0.5 \\
      . &    . &  0.4 &   1 & 0.2 \\
      . &    . &  0.5 & 0.2 &   1 \\
\end{array}  
\right)
\end{eqnarray}

After truncation in some variables (for example $(W, V, X)$ as $-2 \le W \le 1$, $-1 \le V \le 1$, $0 \le X \le 1$),
we apply equation \eqref{sec.moment} to compute the truncated second moment and the inverse covariance matrix as:


\begin{eqnarray}
\bfOmega^* & = &
\left(
\begin{array}{rrrrrr}
	 1.88 &    0.2 &    0.3 &   . & .   \\
    0.2  &  3.45 &   -0.1 &   . & .   \\
    0.3  &  -0.1 &  12.67 & 0.4 & 0.5 \\
      .  &    .  &    0.4 &   1 & 0.2 \\
      .  &    .  &    0.5 & 0.2 &   1 \\
\end{array}  
\right)
\end{eqnarray}
The precision matrix $\bfOmega^*$ after selection differs from 
$\bfOmega$ only in the diagonal elements of $(W, V, X)$. 
From $\bfOmega^*$ we can read how partial correlations between variables 
have changed due to the selection process.
 
Each diagonal element $\bfOmega_{yy}$ of the precision matrix is the inverse of the
partial variance after regressing on all other variables (\cite{Whittaker1990},\allowbreak p.\,143).
Since only those diagonal elements in the precision matrix for the $k \le d$ 
of the truncated variables will change after selection, 
this leads to the idea to just compute these $k$ elements after selection rather 
than the full $k(k+1)/2$ symmetric elements in the truncated covariance matrix and 
applying the Johnson/Kotz formula for the remaining $d-k$ variables. 
However, the inverse partial variance of a scalar $Y$ given the remaining variables $X=\{x_1,\ldots,x_d\} \setminus {y}$
\begin{equation*}
  \bfOmega^*_{yy} = \left[ \Sigma^{*}_{y.X} \right]^{-1} = \left[ \Sigma^{*}_{yy} - \Sigma^{*}_{yX} \Sigma^{*-1}_{XX} \Sigma^{*}_{Xy} \right]^{-1}
\end{equation*}
still requires the truncated covariance results derived in Section 3.

\section{Summary}

In this paper we derived a formula for the first and second moments of the doubly truncated multivariate normal distribution and for their bivariate marginal density. An implementation for both formulas has been made 
available in the R statistics software as part of the \verb|tmvtnorm| package.
We linked our results to the theory of graphical models and 
proved an invariance property for elements of the precision matrix. 
Further research can deal with other types of truncation than we considered (e.g. elliptical).
Another line of research can look at the moments of the doubly truncated multivariate Student-t distribution, which contains the truncated multivariate normal distribution as a special case.

\bibliographystyle{spbasic}
\bibliography{StatPapers}

\end{document}